\documentclass[11pt]{article}
\newcommand{\remove}[1]{}
\newcommand{\ignore}[1]{}

\usepackage{amsmath}
\usepackage{amsfonts}
\usepackage{graphicx}
\usepackage{amssymb}
\usepackage{amsthm}
\usepackage{caption}

\newtheorem{thm}{Theorem}

\newtheorem{lmma}[thm]{Lemma}

\newtheorem{cor}[thm]{Corollary}

\newtheorem{fact}[thm]{Fact}

\theoremstyle{definition}
\newtheorem{defn}[thm]{Definition}
\theoremstyle{remark}

\usepackage[T1]{fontenc}
\usepackage{yfonts}
\usepackage{fullpage}

\title{Secure End-to-End Communication with Optimal Throughput In Unreliable Networks}
\author{Paul Bunn\thanks{Google Inc$.$ {\tt paulbunn@google.com}} \and Rafail Ostrovsky\thanks{UCLA Departments of Computer Science and Mathematics. {\tt rafail@cs.ucla.edu}}}
\date{}
\begin{document}
\maketitle

\begin{small}
\noindent{\bf Abstract.}
We demonstrate the feasibility of end-to-end communication in highly unreliable networks.  Modeling a network as a graph with vertices representing nodes and edges representing the links between them, we consider two forms of unreliability: unpredictable edge-failures, and deliberate deviation from protocol specifications by corrupt and maliciously controlled nodes.

We present a routing protocol for end-to-end communication that is simultaneously resilient to both forms of unreliability.  In particular, we prove that our protocol is {\em secure} against arbitrary actions of the corrupt nodes controlled by a polynomial-time adversary, achieves {\em correctness} (Receiver gets {\em all} of the messages from Sender, in-order and without modification), and enjoys provably optimal throughput performance, as measured using {\em competitive analysis}. Competitive analysis is utilized to provide protocol guarantees again malicious behavior without placing limits on the number of the corrupted nodes in the network.

Furthermore, our protocol does not incur any asymptotic memory overhead as compared to other protocols that are unable to handle malicious interference of corrupt nodes.  In particular, our protocol requires $O(n^2)$ memory per processor, where $n$ is the size of the network.  This represents an $O(n^2)$ improvement over all existing protocols that have been designed for this network model.

\bigskip
\noindent {\bf Keywords:} \hspace*{-1.3pt}Network \hspace*{-.5pt}Routing, \hspace*{-1.25pt}Asynchronous \hspace*{-.5pt}Protocols,  End-to-End Communication, Competitive Analysis, \hspace*{-1pt}Multi-\hspace*{-.75pt}Party \hspace*{-.75pt}Computation in Presence of Dishonest Majority, Fault Localization, Communication Complexity
%
\end{small}
\section{Introduction}
\indent \indent With the immense range of applications and the multitude of networks encountered in practice, there has been an enormous effort to study routing in various settings.  In the present paper, we investigate the feasibility of routing in a network in which neither the nodes nor the links are reliable.

We adopt the same definition of {\em unreliability} (with respect to both the links and the nodes) as was introduced in \cite{paper}.  Namely, for the network links, we do not assume any form of stability: the topology of the network is {\em dynamic} (links may spontaneously fail or come back to life at any time), transmission time across each link may vary from link to link as well as across the same link from one transmission to the next (i.e$.$ {\em asynchronous} edges), and there is no guarantee that there are enough links available (even over time) for communication to even be possible.

Meanwhile, unreliability of network {\em nodes} means that they may actively and maliciously deviate from protocol specifications, attempting to disrupt communication as much as possible.  In particular, a {\em malicious adversary} may corrupt an arbitrary subset of nodes, taking complete control over them and coordinate attacks to interfere with communication between the uncorrupt nodes.

Admittedly, few guarantees can be achieved by any protocol that is forced to operate in networks with so few assumptions.  Indeed, the absence of any assumption on connectivity means that successful routing may be {\em impossible}, for instance if all of the links remain forever inactive.  Therefore, instead of measuring the efficacy of a given protocol in terms of its absolute performance, we will employ {\em competitive analysis} to evaluate protocols: the throughput-performance of a given protocol with respect to the network conditions encountered will be compared to the performance of an {\em ideal} protocol (one that has perfect information regarding the schedule of active/inactive links and corrupt nodes, and makes perfect routing decisions based on this information).

The combination of this strong notion of unreliability together with the use of competitive analysis provides a meaningful mechanism to evaluate routing protocols in networks that demonstrate unreliability in unknown ways.  For example, we are able to compare protocols that route in networks that are susceptible to all of the above forms of unreliability, but e.g$.$ remain stable most of the time with respect to the edges (or alternatively e.g$.$ most of the nodes remain uncorrupted).  Therefore, by allowing networks to exhibit all forms of unreliability, we compromise absolute performance for robustness.  That is, no protocol will route packets quickly through a network that displays all forms of unreliability, but protocols with high competitive-ratio are guaranteed to do as well as possible, regardless of the actual network conditions.

Our approach for developing a routing protocol in unreliable networks will be from a theoretical perspective.  In Section \ref{Model} we provide a formal model for both forms of unreliability and offer definitions of {\em throughput} and {\em security}/{\em correctness} in this model. Section \ref{bach} describes a protocol that is provably optimal with respect to throughput-efficiency (as measured via {\em competitive-analysis}), is provably secure, and requires reasonable memory of internal nodes.  We emphasize that the focus of this paper is on the theoretical feasibility of routing in highly unreliable networks, and in particular no attempt has been made to minimize constants or prototype our protocol in live experiments.
\subsection{Previous Work}\label{previWork}
\indent \indent Development and analysis of routing protocols relies heavily on the assumptions made by the network model.  In this section, we explore various combinations of assumptions that have been made in recent work, highlighting positive and negative results with respect to each network model, emphasizing clearly which assumptions are employed in each case.  Since our work focuses on theoretical results, for space considerations we do not discuss below the vast amount of research and analysis of routing issues for specific network systems encountered in practice, e.g$.$ the Internet.  For example, there are a number of internet routing protocols used in practice that we will not discuss here: TCP, BGP, OSPF, etc. In the presence of adversarial conditions, these protocols often try to achieve ``best effort'' results instead of guaranteeing eventual delivery of all messages.

The amount of research regarding network routing and analysis of routing protocols is extensive, and as such we include only a sketch of the most related work, indicating how their models differ from ours and providing references that offer more detailed descriptions.\\[.2cm]
{\sc End-to-End Communication:} \hspace*{-.1cm} While there is a multitude of problems that involve end-to-end communication (e.g$.$ End-to-End Congestion Control, Path-Measurement, and Admission Control), we discuss here work that consider networks whose only task is to facilitate communication between the Sender and Receiver.  Some of these include a line of work developing the {\em Slide} protocol (the starting point of our protocol): Afek and Gafni and Rosen \cite{AGR}, Awerbuch et al$.$ \cite{AMS},  Afek et al$.$ \cite{Slide}, and Kushilevitz et al$.$ \cite{KO}.  The Slide protocol (and its variants) have been studied in a variety of network settings, including multi-commodity flow (Awerbuch and Leighton \cite{AL}), networks controlled by an online bursty adversary (Aiello et al$.$ \cite{AKOR}), synchronous networks that allow corruption of nodes (Amir et al$.$ \cite{ABO}).  Bunn and Ostrovsky consider in \cite{paper} an identical network model to the one considered in the present paper, and prove a matching upper and lower bound on optimal throughput performance for this model.  However, the mechanisms they employ to handle malicious activity is extremely expensive (in terms of memory); indeed an open problem posed in \cite{paper} was whether a protocol can achieve security against malicious nodes at no extra (asymptotic) cost with respect to memory.  We answer this question affirmatively in this paper, presenting a protocol that reduces memory requirements by a factor of $n^2$ (from $\Theta(n^4)$ to $\Theta(n^2)$, for networks with $n$ nodes).\\[.2cm]
{\sc Fault Detection and Localization Protocols:} \hspace*{-.1cm} There have been a number of papers that explore the possibility of corrupt nodes that deliberately disobey protocol specifications in order to disrupt communication.  In particular, there is a recent line of work that considers a network consisting of a {\em single path} from the sender to the receiver, culminating in the recent work of Barak et al$.$ \cite{BGX} (for further background on fault localization see references therein). In this model, the adversary can corrupt any node (except the sender and receiver) in an adaptive and malicious manner. Since corrupting any node on the path will sever the honest connection between sender and receiver, the goal of a protocol in this model is {\em not} to guarantee that all messages sent are received. Instead, the goal is to {\it detect} faults when they occur and to {\it localize} the fault to a single edge.\\
\hspace*{.5cm}Goldberg et al$.$ \cite{GXBR} show that a protocol's ability to detect faults relies on the assumption that One-Way Functions (OWF) exist, and Barak et al$.$ \cite{BGX} show that the (constant factor) overhead (in terms of communication cost) incurred for utilizing cryptographic tools (such as MACs or Signature Schemes) is mandatory for any fault-localization protocol.  Awerbuch et al$.$ \cite{AHNR} also explore routing in the Byzantine setting, although they do not present a formal treatment of security, and indeed a counter-example that challenges their protocol's security is discussed in the appendix of \cite{BGX}.\\
\hspace*{.5cm}Fault Detection and Localization protocols focus on very restrictive network models (typically synchronous networks with fixed topology and some connectivity assumptions), and throughput-performance is usually not considered when analyzing fault detection/localization protocols.\\[.2cm]
{\sc Competitive Analysis:} \hspace*{-.1cm} Competitive Analysis was first introduced by Sleator and Tarjan \cite{ST} as a mechanism for measuring the {\it worst-case} performance of a protocol, in terms of how badly the given protocol may be out-performed by an {\it off-line} protocol that has access to perfect information.  Recall that a given protocol has {\it competitive ratio} $1/\lambda$ (or is {\it $\lambda$-competitive}) if an ideal off-line protocol has advantage over the given protocol by at most a factor of $\lambda$.\\
\hspace*{.5cm}One place competitive analysis has been used to evaluate performance is the setting of distributed algorithms in asynchronous shared memory computation, including the work of Ajtai et al$.$ \cite{AADW}.  This line of work has a different flavor than the problem considered in the present paper due to the nature of the algorithm being analyzed (computation algorithm versus network routing protocol).  In particular, network topology is not a consideration in this line of work (and malicious deviation of processors is not considered).\\
\hspace*{.5cm}Competitive analysis is a useful tool for evaluating protocols in unreliable networks (e.g.\ asynchronous networks and/or networks with no connectivity guarantees), as it provides best-possible standards (since absolute performance guarantees may be impossible due to the lack of network assumptions).  For a thorough description of competitive analysis, see \cite{BE}.\\[.2cm]
%
{\sc Max-Flow and Multi-Commodity Flow:} \hspace*{-.1cm} The Max-Flow 
and Multi-Commodity Flow models assume synchronous networks with connectivity guarantees and incorruptible nodes (max-flow networks also typically have fixed topology and are {\em global-control}: routing protocols assume nodes can make decisions based on a global-view of the network; as opposed to only knowing what is happening with adjacent links/nodes).  There has been a tremendous amount of work in these areas, see e.g$.$ Leighton et al$.$ \cite{LMPSTT} for a discussion of the two models and a list of results, as well as Awerbuch and Leighton \cite{AL} who show optimal throughput-competitive ratio for the network model in question.\\[.2cm]
{\sc Admission Control and Route Selection:} \hspace*{-.1cm} 
There are numerous models that are concerned with questions of admission control and route selection:  The Asynchronous Transfer Model (see e.g.\ Awerbuch et al$.$ \cite{AAP}), Queuing Theory (see e.g$.$ Borodin and Kleinberg \cite{BK} and Andrews et al$.$ \cite{AAF}), Adversarial Queuing Theory (see e.g$.$ Broder et al$.$ \cite{BFU} and Aiello et al$.$ \cite{AOKR}).  For an extensive discussion about these research areas, see \cite{Plot} and references therein.

%
%
The admission control/route selection model assumes synchronous communication and incorruptible nodes and makes connectivity/liveness guarantees.  Among the other options (fixed or dynamic topology, global or local control), each combination has been considered by various authors, see the above reference for further details and results within each specific model.
\subsection{Our Results} \label{ourResultz}
\indent \indent We consider the feasibility of end-to-end routing in highly unreliable networks, where unreliability is encountered with respect to both the network's edges and its nodes.  In particular, we consider asynchronous networks with dynamic topology and no connectivity guarantees; comprised of corruptible nodes that may deviate from protocol specifications in a deliberately malicious manner.

We present a protocol that routes effectively in this network setting, utilizing standard cryptographic tools to guarantee {\em correctness} with low memory burden per node.  We use {\em competitive-analysis} to evaluate the throughput-efficiency of our protocol, and demonstrate that our protocol achieves optimal {\em throughput}.  Our protocol therefore represents a constructive proof of the following theorem (see Section \ref{Model} for definitions of our network model and the above terms):
\begin{thm} \label{mtn} Assuming Public-Key Infrastructure and the existence of a group-homomorphic encryption scheme, there exists a routing protocol that achieves correctness and optimal competitive-ratio $1/n$ in a distributed asynchronous network with bounded memory and dynamic topology (and no connectivity assumptions), even if an arbitrary subset of malicious nodes deliberately disobey the protocol specifications in order to disrupt communication as much as possible.
%
\end{thm}
%
As mentioned in Section \ref{previWork}, our protocol solves an open problem from \cite{paper}, which was to provide provable security (while maintaining optimal throughput) at {\bf no} additional cost (in terms of required processor memory) over protocols that do not provide security against corrupt nodes. Our protocol utilizes novel techniques to achieve exactly this: the memory burden is reduced from\footnote{These bounds ignore the cost of security parameter $k$ and bandwidth parameter $P$, which are treated as constants. Including them explicitly would yield memory costs of $\Theta(kPn^4)$ for the protocol of \cite{paper} versus $\Theta(kPn^2)$ here.} $\Theta(n^4)$ to $\Theta(n^2)$, which matches the memory requirements of a corresponding (insecure) protocol of \cite{Slide}. We provide here a brief overview of the new insights that enabled us to achieve this reduction.

We begin by describing why the $\Theta(n^4)$ bits of memory per node was required in \cite{paper} to ensure security. Consider the \textsf{packet-replacement} adversarial strategy, where corrupt nodes replace new packets they receive with duplicate copies of old packets that they have already transferred, thereby effectively deleting all new packets the Sender inserts. The protocol of \cite{paper} protected against this strategy by having each node maintain a signed transaction with each of its neighbors, recording the number of times {\em every} packet was passed between them. While this approach ensures that a node performing packet-replacement will be caught, it is extremely costly in terms of required memory: Each node has to remember, for every packet $p$ it encountered, the number of times it sent/received $p$ along each of its adjacent edges. For networks with $n$ nodes, since there were $\Theta(n^3)$ relevant packets\footnote{The $n^3$ appearing here is {\em not} a bound on the size of the input stream of packets (which can be any arbitrarily large polynomial in $n$); it is an upper-bound on the number of packets being stored by the internal nodes at any time.} and a node may have $\Theta(n)$ neighbors, the memory burden of storing this transaction information was $\Theta(n^4)$. Not only did this large memory complexity mean the protocol of \cite{paper} was unlikely to be feasibly implemented in practice, it was also the case that the cost of $n^4$ for storing the transaction history (for the purpose of identifying corrupt behavior) far out-weighed the per-node memory costs of the data packets being transferred ($n^2$), so the memory resources were being consumed by network {\em monitoring} as opposed to {\em routing}.

The present paper overcomes both of these issues, reducing the overall memory burden to $\Theta(n^2)$, as well as allocating the majority of resources to routing instead of monitoring. In order to achieve this, we had to abandon the idea of tracking each individual packet, and develop a novel technique to address packet-replacement. We began by generalizing the per-packet tracking of \cite{paper} as follows: We partition the $D = \Theta(n^3)$ packets to be sent into $K$ sets $\{\mathcal{S}_1, \dots, \mathcal{S}_K\}$ (we will optimize for the value of $K(k)$, which depends on the security-parameter $k$, in Section \ref{bach}), and then we have nodes record transaction information with their neighbors on a per-{\em set} basis rather than a per-{\em packet} basis. Namely, nodes maintain $K$ counters of how many packets in each set they have transferred with each neighbor, so that if a packet $p \in \mathcal{S}$ is transferred between two nodes, the nodes increment a counter for set $\mathcal{S}$. In this way, if a malicious node replaces a packet $p \in \mathcal{S}$ with a packet $p' \in \mathcal{S}'$, the per-set counters will help in detecting this if $\mathcal{S} \neq \mathcal{S}'$.

With this generalization, we see that as $K$ varies in $[1, D]$, there is a trade-off in the memory burden of storing the transactions versus the probability of protecting against packet-replacement: the smaller $K$ is, the lower the per-node memory burden but the higher the probability a node performing packet-replacement can get away with it. The primary technical achievement of this paper was in developing a mechanism that guarantees that {\em any} packet-replacement strategy performed by malicious node(s) will succeed only with negligible probability, even for small values of $K$.

We achieve this by first using error-correction to ensure that our protocol is robust enough to handle minor amounts of packet-replacement and still transmit messages, so that in order to impede communication via the packet-replacement strategy, a large number of packets must be replaced. Next, we observe that if a malicious node replaces a packet $p \in \mathcal{S}$ with $p' \in \mathcal{S}'$, then if the choices of $p$ and $p'$ are uniformly random (among the $D$ total packets), then the probability that $\mathcal{S} = \mathcal{S}'$ is roughly $1/K$. By using cryptography, we are able to obfuscate the partitioning of packets into sets in a manner that is invisible to all nodes except the Sender, and we demonstrate how this reduces {\em any} adversarial strategy of packet-replacement to the uniform case of replacing one packet with a randomly chosen second packet. With this reduction in hand, it becomes a straightforward probabilistic analysis for choosing an appropriate value for the parameter $K$ so as to minimize memory burden and still guarantee (with negligible probability of error) that packet-replacement will be detected. Details of the protocol and this analysis can be found in Section \ref{bach}.
\section{The Model} \label{Model}
\indent \indent In this section, we describe formally the model in which we will be analyzing routing protocols.  The network is viewed as a graph $G$ with $n$ vertices (or {\em nodes}), two of which are designated as the {\em Sender} and {\em Receiver}.  The Sender has a stream of messages $\{m_1, m_2, \dots \}$ that it wishes to transmit through the network to the Receiver.

For ease of discussion, we assume that all edges in the network have a fixed bandwidth/capacity, and that this quantity is the same for all edges.  We emphasize that this assumption does not restrict the validity of our claims in a more general model allowing varying bandwidths, but is only made for ease of exposition. We will use the following terminology throughout this paper (see protocol description in Section \ref{lastDesc} for more explanation of these terms an how they are used):

\begin{defn} \label{messageDef} Let $P$ denote the bandwidth (e.g$.$ in bits) of each edge. A \textsf{packet} (of size $\leq P$) will refer to any bundle of information sent across an edge. A \textsf{message} refers to one of the Sender's input $m_i$, and we assume without loss of generality that each message is comprised of $\Theta(kPn^3)$ bits ($k$ is the security parameter; see below). A (message) \textsf{codeword} refers to an encoded message, which will be partitioned into \textsf{codeword parcels}, whose size is small enough such that one codeword parcel (plus some control information) fits in the bandwidth of an edge $P$. More generally, we will refer to the various components of a packet as \textsf{parcels}. A (message) \textsf{transmission} consists of the rounds required to send a single codeword from Sender to Receiver.
\end{defn}

We model {\em asynchronicity} via an \textsf{edge-scheduling adversary} $\mathcal{A}$ that controls edges as follows.  A \textsf{round} consists of a single edge $E(u,v)$ (chosen by the adversary) being \textsf{activated}:
	\begin{itemize}
	\item[1.] If $\mathcal{A}$ has at least one packet from $u$ to be sent to $v$, then $\mathcal{A}$ delivers exactly one of them (of $\mathcal{A}$'s choosing) to $v$; the same is done for one packet from $v$ to $u$
	\item[2.] After seeing the delivered packet, $u$ chooses the next packet to send $v$ and gives it to $\mathcal{A}$ ($v$ does the same for $u$); $\mathcal{A}$ will store these until the next round that $E(u,v)$ is activated
	\end{itemize}
If $u$ does not have a packet it wishes to send $v$ in Step (2), then $u$ can choose to send nothing.  Alternatively, $u$ may send multiple packets to $\mathcal{A}$ in Step 2, but only one of these packets (of $\mathcal{A}$'s choosing) gets delivered in Step 2 of the next round $E(u,v)$ is activated. The Adversary does not send anything to $v$ in Step (1) if it is not storing a packet from $u$ to $v$ during round $E(u,v)$.

\begin{defn} A packet will be said to be in an \textsf{outstanding request} if $u$ has sent the packet to $\mathcal{A}$ as in Step (2) of some round, but that packet has not yet been delivered by $\mathcal{A}$.
\end{defn}

Aside from obeying the above specified rules, we place no additional restriction on the edge-scheduling adversary.  In other words, it may activate whatever edges it likes (this models the fact our network makes \textsf{no connectivity assumptions}), wait indefinitely long between activating the same edge twice (modeling both the \textsf{dynamic} and \textsf{asynchronous} features of our network), and do anything else it likes (so long as it respects steps (1) and (2) above each time it activates an edge) in attempt to hinder the performance of a routing protocol.

In addition to the edge-scheduling adversary, our network model also allows for a polynomially bounded (in number of nodes $n$ and a security parameter $k$) \textsf{node-controlling adversary} to corrupt the nodes of the network.  The node-controlling adversary is \textsf{malicious}, meaning that it can take complete control over the nodes it corrupts and force them to deviate from any protocol in whatever manner it likes.  We further assume that the node-controlling adversary is \textsf{adaptive}, which means it can corrupt nodes at any stage of the protocol, deciding which nodes to corrupt based on what it has observed thus far.  We do not impose any ``access-structure'' limitations on the node-controlling adversary:  it may corrupt any nodes it likes (although if the Sender and/or Receiver is corrupt, secure routing between them is impossible).  We say a routing protocol is \textsf{correct} (or \textsf{secure}) if the messages reach the Receiver in-order and unaltered.

The separation of the (edge-scheduling and node-controlling) adversaries into two distinct entities is solely for conceptual purposes to emphasize the nature of unreliability in the edges versus the nodes.  For ease of discussion, we will often refer to a single adversary that represents the combined efforts of the edge-scheduling and node-controlling adversaries.

Finally, our network model is \textsf{on-line} and \textsf{distributed}, in that we do not assume that the nodes have access to any information (including future knowledge of the adversary's schedule of activated edges) aside from the packets they receive during a round they are a part of.  Also, we insist that nodes have \textsf{bounded memory}\footnote{For simplicity, we assume all nodes have the same memory bound (which may be a function of the number of nodes $n$ and security parameter $k$), although our argument can be readily extended to handle the more general case.} which is at least $\Omega(n^2)$.

Our mechanism for evaluating the \textsf{throughput} performance of protocols in this network model will be as follows: Let $f^{\mathcal{A}}_{\mathcal{P}}: \mathbb{N} \rightarrow \mathbb{N}$ be a function that measures, for a given protocol $\mathcal{P}$ and adversary $\mathcal{A}$, the number of messages that the Receiver has received as a function of the number of rounds that have passed.  Note that in this paper, we will consider only deterministic protocols, so $f^{\mathcal{A}}_{\mathcal{P}}$ is well-defined.  The function $f^{\mathcal{A}}_{\mathcal{P}}$ formalizes our notion of throughput.

We utilize competitive analysis to gauge the throughput-performance of a given protocol against all possible competing protocols:
\begin{defn} \label{tdef} We say that a protocol $\mathcal{P}$ has \textsf{competitive-ratio} $\boldsymbol{1/\lambda}$ (respectively is $\boldsymbol{\lambda}$-\textsf{competitive}) if there exists a constant $c$ and function {\em g(n, C)} (where $C$ is the memory bound per node) such that for all possible adversaries $\mathcal{A}$ and for all $\mathtt{x} \in \mathbb{N}$, the following holds for all protocols $\mathcal{P'}$:\footnote{$\lambda$ is taken as the infimum of all values that satisfy \eqref{LambdaDef}, and is typically a function of network size $n$.}
	\begin{equation} \label{LambdaDef}
	f^{\mathcal{A}}_{\mathcal{P}'}(\mathtt{x}) \thickspace \leq \thickspace (c \thickspace \cdot \thickspace \lambda) \thickspace \cdot f^{\mathcal{A}}_{\mathcal{P}}(\mathtt{x}) \thickspace + \thickspace g(n,C)
	\end{equation}
\end{defn}
\noindent Note that while $g$ may depend on the network size $n$ and the bounds placed on processor memory $C$, both $g$ and $c$ are {\em independent} of the round $\mathtt{x}$ and the choice of adversary $\mathcal{A}$.  Also, equation \eqref{LambdaDef} is only required to hold for protocols $\mathcal{P'}$ that never utilize a corrupt node once it has been corrupted.

We assume a Public-Key Infrastructure (PKI) that allows digital signatures.  In particular, before the protocol begins we choose a security parameter sufficiently large and run a key generation algorithm for a digital signature scheme, producing $n = |G|$ (secret key, verification key) pairs $(sk_u, vk_u)$.  As output to the key generation, each processor $u \in G$ is given its own private signing key $sk_u$ and a list of all $n$ signature verification keys $vk_{v}$ for all nodes $v \in G$.  In particular, this allows the Sender and Receiver to sign messages to each other that cannot be forged (except with negligible probability in the security parameter) by any other node in the system.

We also assume the existence of a {\em group-homomorphic} encryption scheme $\mathcal{E}$ on a group $\mathcal{G}$:\footnote{$|\mathcal{G}|$ should be larger than the total number of codeword parcel transfers (during any transmission) between two nodes, when at least one of the nodes is {\em honest}. $|\mathcal{G}| \in \Omega(kn^4)$ is sufficient (see protocol description in Section \ref{lastDesc}).}
	\begin{equation*}
	\mathcal{E} : \mathcal{G} \rightarrow \mathcal{H}, \quad \mbox{with} \quad \mathcal{E}(g_1 \circ_{\mathcal{G}} \thickspace g_2) \thickspace = \thickspace \mathcal{E}(g_1) \circ_{\mathcal{H}} \thickspace \mathcal{E}(g_2),
	\end{equation*}
where $\circ_{\mathcal{G}}$ (respectively $\circ_{\mathcal{H}}$) represents the group operation on $\mathcal{G}$ (respectively $\mathcal{H}$). To simplify the exposition, in what follows we will assume $\mathcal{G} = \mathbb{Z}_N$ for sufficiently large $N$.  We note that such a scheme exists under most of the commonly used cryptographic assumptions, including {\em factoring} \cite{ou}, {\em discrete log} \cite{elG}, {\em quadratic residuosity} \cite{gm}, and {\em subgroup decision problem} \cite{bgn}.  We extend our encryption scheme to $\mathbb{Z}_N^K$ in the natural way:
	\begin{equation*}
	\mathcal{E}\hspace*{-2.2pt}: \mathbb{Z}_N \times \dots \times \mathbb{Z}_N \thickspace \rightarrow \thickspace \mathcal{H} \times \dots \times \mathcal{H} \qquad \hspace*{-5.1pt}\mbox{via} \qquad \hspace*{-5.1pt} \mathcal{E}(g_1, \dots, g_K) \hspace*{-1pt}:= (\mathcal{E}(g_1), \dots, \mathcal{E}(g_K))
	\end{equation*}

Finally, we assume that internal nodes have capacity $C \in \Omega(Pn^2)$ (and in particular $C \geq 24Pn^2$), and that $P = \Omega(k^2+\log n)$ (where $P$ is the capacity of each edge and $k$ is the security parameter for the encryption scheme).
\section{Routing Protocol} \label{bach}
\indent \indent In this section we present a routing protocol that enjoys competitive-ratio $1/n$ with respect to throughput (which is optimal, see \cite{paper}) in networks modelled as in Section \ref{Model}.
\subsection{Description of the Routing Protocol}\label{lastDesc}
\indent \indent The starting point of our protocol will be the Slide Protocol, introduced by Afek et at.\ \cite{AGR}, and further developed in a series of works: \cite{AMS}, \cite{Slide}, \cite{KO}, \cite{ABO}, and \cite{paper}.  The original Slide protocol assumes that nodes have buffers (viewed as stacks) able to store $C=\Theta(n^2)$ packets at any time, and simply put, it calls for a node $u$ to send a packet to node $v$ across an activated edge $E(u,v)$ if $v$ is storing fewer packets in its buffer than $u$.

The Slide protocol is robust in its ability to handle edge-failures (modelled here via the edge-scheduling adversary). This robustness is achieved via the use of {\em error-correction} to account for packets that get stuck in the buffer of a node that became isolated from the rest of the network due to edge-failures.  In particular, each message is expanded into a codeword, which is then partitioned into $D:=knC/\lambda$ parcels $\{p_i\}$, where $\lambda$ is the tolerable error rate and $k$ is the security parameter. Recall from Definition \ref{messageDef} that messages have size $|m| = O(kPn^3)$, and in particular they are small enough so that a codeword parcel (plus some control information of $\Theta(k^2 + \log n)$ bits, see below) can be transmitted in a single round; i.e$.$ $|p_i| \hspace*{-1pt}= \hspace*{-1pt}|m| / D \leq \hspace*{-1pt}P \hspace*{-0.5pt}- \hspace*{-0.5pt}\Theta(k^2 + \log n)$. The Receiver can decode the codeword and obtain the original message block provided he receives $(1-\lambda)D$ codeword parcels.

Our protocol modifies the original Slide protocol to provide security against a node-controlling adversary at no additional (asymptotic) cost: we achieve optimal throughput (competitive-ratio $1/n$), and the memory per internal node is within a factor of two of the memory requirement of the original Slide protocol. We obtain security against malicious nodes by including extra \textsf{control information} (described below) with each packet transfer, and by having nodes sign all communications. As mentioned in Section \ref{ourResultz}, our protocol closely resembles that of \cite{paper}, except we have changed the nature of the control information so that our protocol is able to provide security without incurring any extra (asymptotic) memory costs.

The rules governing codeword parcel transfers are the same as in Slide (a node should send a parcel to its neighbor if it is currently storing more parcels), except for the caveat that nodes should not transfer any codeword parcels before they have received all of the Sender's alert parcels for the current codeword transmission (explained below); and nodes should never transfer codeword parcels with blacklisted or eliminated nodes (see below). Also, all communication is signed to ensure authenticity: the Sender signs all packets upon inserting them into the network (so a node can verify the authenticity of each codeword parcel), and control information between two nodes is signed by both nodes. In the following subsections, we present our protocol by describing the control information, routing rules, blacklist, and overall execution of the protocol.
%

\vspace*{12pt}\noindent\textbf{Control Information.} When a node is selecting a codeword parcel to send to an adjacent node, he will have room to bundle four parcels of control information plus the codeword parcel in a single packet.\footnote{The size of the codeword parcels are deliberately set to ensure a packet has room to fit these five parcels, since codeword parcels were designed to have size $P - \Theta(k^2 + \log n)$.}  In particular, every packet transferred will include one parcel of each of the four categories of \textsf{control information} (how the node selects the specific parcel within each category will be explained in the Routing Rules below):
	\begin{enumerate}
	\item \textsf{Sender/Receiver Alerts}. The Sender's alert consists of up to $2n$ parcels, all time-stamped with the index of the present codeword transmission. The first of these indicates the status (S1 or F2-F4, see below) of the previous transmission; and the next $n$ parcels give the time-stamp of the most recent (up to) $n$ transmissions that {\em failed} (F2-F4). The final $n-1$ parcels are for each of the nodes (excluding the Sender), indicating if that node is blacklisted or eliminated (see below), and if so the transmission this happened. The Receiver's alert consists of a single parcel indicating either that the Receiver successfully decoded the current codeword, or that it has received inconsistent potential information (see below).

	\item \textsf{Potential Information}. Each node $u$ maintains (in a single parcel) up-to-date information about its potential drop $\Phi_u$ (see Definition \ref{potDropp}) for the current codeword transmission.

	\item \textsf{Status Information}. For each of its neighbors, a node will maintain (in a single parcel) up-to-date information regarding all codeword parcel transfers with that neighbor for the current codeword transmission. More specifically, this information contains the net {\em potential drop} $\Phi_{u,v}$ and {\em obfuscated count} $\Psi_{u,v}$ across each adjacent edge $E(u,v)$ (see Definitions \ref{potDropp} and \ref{obfCount} below).
	
	\item \textsf{Testimonies}. At the end of a codeword transmission $\mathtt{T}$, a node will have one final (current as of the end of the transmission) status parcel $(\Phi_u, \Phi_{u,v}, \Phi_{v,u}, \Psi_{u,v}, \Psi_{v,u})$ for each neighbor. If the node later (in a future transmission) learns that $\mathtt{T}$ {\em failed} (see below), then these $n-1$ status parcels become the node's \textsf{testimony} for transmission $\mathtt{T}$. Since nodes do not participate in routing codeword parcels until the Sender has its testimony (see blacklist below), each node will only ever have at most one transmission $\mathtt{T}$ for which it needs to remember its own testimony; thus, at any time, there are at most $(n-1)^2$ testimony parcels in the network.\footnote{Since it is the Sender who ultimately collects testimonies to identify malicious behavior, his testimony parcels need not be stored or transferred to any node other than itself.}
	\end{enumerate}
Each of these types of control information serve a separate function. The Sender's alert parcels mark the start of a new codeword transmission, and relay information about which previous transmissions failed (so that nodes know which testimony parcels to store and transfer) and which nodes are blacklisted/eliminated (see below). The Receiver's alert parcel marks the end of a transmission, which either indicates a \textsf{successful} transmission (Receiver could decode the codeword) or that the transmission failed due to inconsistencies in potential differences (see below); ultimately it is the Sender who will process the Receiver's alert parcel and determine the next steps.

Potential parcels are ultimately used by the Receiver, who will use them to determine if there are any inconsistencies, and if so form a Receiver alert indicating a failed transmission (see \eqref{potEqn} below).

Unlike all other control information, status parcels are {\em not} transferred through the network (ultimately to Sender or Receiver), but rather are only kept locally and transferred between the two nodes for which the status parcel is keeping up-to-date records for the current transmission. Testimony parcels are ultimately used by the Sender to identify a corrupt node.

We now formalize these concepts with a handful of definitions.
\begin{defn} \label{htDef} The \textsf{height} $\boldsymbol{H_u}$ of an internal node $u$ is the number of codeword parcels $u$ is currently storing in its buffer (including those in outstanding requests, of which $u$ is maintaining a copy). The height of the Sender is defined to be the constant $C$ (the capacity of an internal node's buffer); and the Receiver's height is defined to be zero.
\end{defn}
\begin{defn} \label{potDropp} The \textsf{potential difference} $\boldsymbol{\phi_{u,v}}$ of two nodes $u$ and $v$ is the difference in their heights (always measured as a positive quantity): $\phi_{u,v} := | H_u - H_v |$.
The \textsf{directional potential drop} $\boldsymbol{\Phi_{u,v}}$ \textsf{over an edge} $\boldsymbol{E(u,v)}$ will be the sum of the potential differences for the rounds when $u$ transferred $v$ a codeword parcel:\footnote{Formally, for a packet transferred during Step (1) of an edge activation, the heights used to compute the potential difference are {\em not} the current heights of the nodes, but rather the heights each of the nodes had the {\em previous} time the edge was activated. See Figure \ref{briefDisc}, in which these heights are denoted as $H_v$ and $H_{old}$.}
	\begin{equation}
	\Phi_{u,v} := \sum_{u \rightarrow v} \phi_{u,v}
	\end{equation}
The \textsf{potential drop over an edge} $\boldsymbol{E(u,v)}$ will be the sum of the directional potential drops across the edge: $\Phi_{u,v} + \Phi_{v,u}$. The \textsf{potential drop at a node $\boldsymbol{u}$} $\boldsymbol{\Phi_u}$ will be the sum of the potential drops over all its adjacent edges:
	\begin{equation}\label{paperW}
	\Phi_u := \sum_{v \in G} (\Phi_{u,v} + \Phi_{v,u})
	\end{equation}
\end{defn}

Recall from Section \ref{Model} the existence of a homomorphic encryption scheme $\mathcal{E}$ on $\mathbb{Z}_N^K$.  At the start of each codeword transmission, the Sender randomly partitions the $D$ codeword parcels into $K:=k$ sets, making a uniform random choice for each parcel. Define the distribution $\chi_{_{\mathtt{T}}}\hspace*{-2pt}: D \rightarrow \mathbb{Z}_N^K$ (which depends on the codeword transmission $\mathtt{T}$) to represent these assignments; i.e$.$ if parcel $p$ has been assigned to the $i^{th}$ set, then $\chi(p)$ is the unit vector in $\mathbb{Z}_N^K$ with a `1' in the $i^{th}$ coordinate.  Note that only the Sender knows $\chi$, and it will remain obfuscated from all internal nodes, as the only information they will ever see are the encrypted values $\mathcal{E}(\chi(p))$ (which are computed by the Sender and bundled in the same packet as the codeword parcel, so that $(p, \mathcal{E}(\chi(p))$) will travel in the same packet as it is transferred through the network to the Receiver).
\begin{defn} \label{obfCount} The \textsf{directed obfuscated count} $\boldsymbol{\Psi_{u,v}}$ between two nodes is an (encrypted) $K$-tuple, in which the $i^{th}$ coordinate represents the number of codeword parcels $p$ with $\chi(p) = i$ that have been transferred from $u$ to $v$. Since $p$ and $\mathcal{E}(\chi(p))$ are always passed together in a single packet, the current directed obfuscated count can be computed by any internal node along any of its adjacent edges as:
	\begin{equation}\label{eObfCount}
	\Psi_{u,v} := \sum_{p \in \mathcal{P}_{u,v}} \hspace*{-6pt}\mathcal{E}(\chi(p)),
	\end{equation}
where $\mathcal{P}_{\hspace*{-2pt}u,v}$ \hspace*{-1pt}denotes the multiset of codeword parcels transferred from \hspace*{-1pt}$u$ to $v$. The \textsf{obfuscated count} $\boldsymbol{\Psi_{\hspace*{-1pt}u}}$ \hspace*{-1pt}\textsf{at $\boldsymbol{u}$} is: 
	\begin{equation}\label{motart}
	\Psi_u := \sum_{p \in u} \mathcal{E}(\chi(p)),
	\end{equation}
where the sum is taken over the (current codeword) parcels $p$ that $u$ is storing at the end of the current transmission. Notice that the homomorphic properties of the encryption scheme $\mathcal{E}$ allow the nodes to compute the right-hand-side of \eqref{eObfCount} and \eqref{motart}.
\end{defn}
\vspace*{12pt}\noindent\textbf{Routing Rules.} Figure \ref{briefDisc} gives a succinct description of a node's instructions for when it is part of an activated edge.

\begin{figure}[h!t]
\vspace*{2mm}
\begin{center}
\framebox{\footnotesize
\begin{minipage}[b]{6.0in}
\begin{footnotesize}
\begin{tabbing}
l\=al\=aaaaaaaaaaaaaaaal\=aaaaaaaaaaaaaaaaaaaa\=aaaaaaaaaaaa\=aaaaaaaaaaa\=aaaaaaaaaaaaaaaaaaaaaaaa\=aa\=a \kill
\\[-13.5pt]\textbf{Routing Rules for node $u$ along $E(u,v)$}\\[-1.01pt]
\# Notation: $(H_{old}, p_{old}, \mathcal{E}(\chi(p_{old})))$ denotes prev$.$ ht$.$ and codeword parcel $u$ sent $v$;\hspace*{-0.3pt}\\[-1.01pt]
\textsf{Input (Received via $\mathcal{A}$)}:\\[-2.2pt]
\> Height $H_v$ of $v$, codeword parcel \hspace*{-1pt}$p$ and $\mathcal{E}(\chi(p))$,\\[-1.01pt]
\> Control Information: alert parcel, status parcel, potential parcel, testimony parcel\\[2.01pt]
\textsf{DO}:\\[-2.751pt]
\>Verify status parcel and all signatures are valid, if not, Skip to \textbf{Send Next Packet}\hspace*{-2.3pt}\\[-1.01pt]
\>Store alert parcel, potential parcel, and testimony parcel\\[-1.01pt]
\>If \hspace*{-1pt}$u$ or \hspace*{-1pt}$v$ is blacklisted or eliminated, \hspace*{-1pt}or \hspace*{-1pt}$u$ hasn't rec'\hspace*{-1pt}d all parcels from Sender's alert:\hspace*{-2.3pt}\\[-1.01pt]
\>\>Skip \hspace*{-.5pt}to \hspace*{-.5pt}\textbf{Send \hspace*{-1pt}Next \hspace*{-1pt}Packet}\hspace*{-1pt}\\[-1.01pt]
\>If $u$ is the Sender and $H_v<C+2n-C/2n$:\\[-1.01pt]
\>\>\textbf{Insert} $\boldsymbol{p_{old}}$: \>Ignore $p$, Delete $p_{old}$, Update $\Phi_{u,v}$, $\Psi_{u,v}$, and $\Phi_u$\\[-1.01pt]
\>If $u$ is the Receiver and $H_v > C/2n - 2n$:\\[-1.01pt]
\>\>\textbf{Receive} $\boldsymbol{p}$: \>Store $p$, Update $\Phi_{v,u}$, $\Psi_{v,u}$, and $\Phi_u$\\[-1.01pt]
\>If $u$ is not Sender or Receiver and $H_{old} > H_v - 2n + C/2n$:\\[-1.01pt]
\>\>\textbf{Send} $\boldsymbol{p_{old}}$: \>Ignore $p$, Delete $p_{old}$, Update $\Phi_{u,v}$, $\Psi_{u,v}$, and $\Phi_u$\\[-1.01pt]
\>If $u$ is not Sender or Receiver and $H_{old} < H_v + 2n - C/2n$:\\[-1.01pt]
\>\>\textbf{Receive} $\boldsymbol{p}$: \>Store $p$ (and keep $p_{old}$), Update $\Phi_{v,u}$, $\Psi_{v,u}$, and $\Phi_u$\\[-1.01pt]
\>\textbf{Send Next Packet}\\[-13pt]
\end{tabbing}
\end{footnotesize}
\end{minipage}}
\end{center}
\vspace{-15pt}
\caption{\small{Succinct Description of Packet Transfer Rules of Our Protocol}}
\label{briefDisc}
\vspace*{-4mm}
\end{figure}

Recall that Step 2 of an activated edge calls for node $u$ to send a packet to $\mathcal{A}$ that it wishes to deliver to $v$ next time $E(u,v)$ is activated. The following rules explain how $u$ decides which data to include in the packet (see {\em Send Next Packet} in Figure \ref{briefDisc}):
	\begin{enumerate}
	\item Current Height $H_u$ (see Definition \ref{htDef})

	\item Codeword Parcel\footnote{If at any time $\Phi_u > kCD$, then $u$ stops transferring codeword parcels altogether (sending a special indicator $\bot$ for its height $H$ so that other nodes know not to exchange codeword parcels with $u$). Since each codeword parcel transfer corresponds in an {\em increase} of at least $\hspace*{-.75pt}C/\hspace*{-.75pt}n$\hspace*{1.5pt}-\hspace*{1.5pt}$2n = \hspace*{-.75pt}\Theta(n)$ to $\Phi_u$, this ensures honest nodes will transfer at most $O(k\hspace*{-1pt}^2n\hspace*{-1pt}^4\hspace*{-1pt})$ codeword parcels\hspace*{-1pt}, and also bounds the number of distinct signatures from $u$ per transmission by $O(k\hspace*{-1pt}^2n\hspace*{-1pt}^4\hspace*{-1pt})\hspace*{-1pt}$).} \hspace*{-2pt}$(p, \hspace*{-.5pt}\mathcal{E}\hspace*{-1pt}(\hspace*{-1pt}\chi(p\hspace*{-1pt})\hspace*{-1pt})$: \hspace*{-1.5pt}Chosen randomly among those not already in outstanding request
	
	\item Control Information: Send up to one parcel of each type of control information, selected as follows. Let $N_1, N_2, \dots, N_n$ denote the nodes and $A = A(u,v)$ the number of times $E(u,v)$ has been activated so far.\footnote{Note that in order to determine which control information parcels to send, $u$ will have to store the following along each edge: a) how many times $A = A(u,v)$ edge $E(u,v)$ has been activated in the current transmission; b) the previous alert parcel $u$ sent to $v$; c) for each node $N_i$, the previous testimony packet of $N_i$'s that $u$ sent to $v$.} Then $u$ includes the following parcels of control information:
		\begin{itemize}\setlength{\itemsep}{0pt} \setlength{\parskip}{0pt} \setlength{\parsep}{0pt}\vspace*{-4pt}
		\item[-] Sender/Receiver Alert: Receiver's alert (if $u$ has it); otherwise next Sender alert parcel
		\item[-] Status Information: ($\Phi_{u,v}, \Phi_{v,u}, \Psi_{u,v}, \Psi_{v,u})$
		\item[-] Potential Information: Let $i \equiv A \hspace*{3pt}(\hspace*{-7pt}\mod n)$; select parcel $\Phi_{N_i}$ (if $u$ has it)
		\item[-] Testimony: Let $i \equiv A \hspace*{3pt}(\hspace*{-7pt}\mod n)$; select next testimony parcel of $N_i$'s (if $u$ has it)
		\end{itemize}
	\end{enumerate}
\vspace*{9pt}\noindent\textbf{Putting It All Together.} The above sections have focused on what happens in a single codeword transmission. We now give an overview of the entire protocol, including how/when the transmission of one message codeword ends and the next begins, and how the blacklist is used to regulate the influence of potentially corrupt nodes.

The end of each transmission is marked by one of the following four events:
	\begin{enumerate}\setlength{\itemsep}{0pt} \setlength{\parskip}{0pt} \setlength{\parsep}{0pt}\vspace*{-5pt}
	\item[]\hspace*{-5pt}\textsf{S1} \hspace*{2pt}Sender gets Receiver alert indicating successful decoding of the current codeword
	\item[]\hspace*{-5pt}\textsf{F2} \hspace*{2pt}Sender gets Receiver alert indicating inconsistencies in potential differences
	\item[]\hspace*{-5pt}\textsf{F3} \hspace*{2pt}Sender inserted all (current) codeword parcels (and S1 did not occur)
	\item[]\hspace*{-5pt}\textsf{F4} \hspace*{2pt}Sender is able to identify a corrupt node
	\end{enumerate}
In the case of S1, the codeword was delivered \textsf{successfully}, and the Sender will begin the next codeword transmission. In the case of F4, the Sender will re-start a new transmission for the same codeword and indicate (in the Sender alert) that the corrupt node has been {\em eliminated} from the network. All nodes are forbidden transferring codeword packets with eliminated nodes (a node will always know the list of eliminated nodes via the Sender's alert before it has any codeword parcels to transfer; see Figure \ref{briefDisc}).

Cases F2 and F3 correspond to \textsf{failed} attempts to transfer the current codeword due to corrupt nodes disobeying protocol rules.  When a transmission $\mathtt{T}$ fails as in cases F2 and F3, the nodes (excluding the Sender) that are not already on a blacklist or eliminated will be put on transmission $\mathtt{T}$'s \textsf{blacklist}; more generally, we will say a node is on the \textsf{blacklist} (or \textsf{blacklisted}) if there is some transmission $\mathtt{T}$ for which the node is on $\mathtt{T}$'s blacklist. Thus after a transmission fails as in F2 or F3, every node (except for the Sender) is either eliminated or blacklisted. As indicated in the Routing Rules of Figure \ref{briefDisc}, packets sent to/from a blacklisted node will not contain a codeword parcel (just control information). A node is removed from the blacklist either when the Sender has received its complete testimony, or when a node is eliminated (whichever happens first). In the former case, the Sender will add a new parcel to the Sender alert, simply indicating the node has been removed from the blacklist.  In the latter case, the Sender will immediately end the transmission as in F4 (described above).  We will say a (non-eliminated) node \textsf{participated} in a transmission if that node was not on the blacklist for at least one round of the transmission.

When one transmission ends as in S1 or F2-F4 and before the next begins, the Sender constructs the (up to) $2n$ parcels of the Sender alert. As indicated in the Routing Rules of Figure \ref{briefDisc}, nodes do not begin transferring codeword parcels until they have received the complete Sender alert. If a node learns that it is on the blacklist for some transmission $\mathtt{T}$ (note that by construction, $\mathtt{T}$ will necessarily be the previous transmission that the node participated in), then the node constructs its testimony for $\mathtt{T}$, which is simply the final values of its status parcels along all its adjacent edges:
	\begin{equation*}
	\mbox{Node $u$'s testimony for a Transmission $\mathtt{T}$} \thickspace \hspace*{-0.2pt} = \hspace*{-0.2pt}\thickspace \{ \Phi_{u,v}, \Phi_{v,u}, \Psi_u, \Psi_{u,v}, \Psi_{v,u} \}_{v \in G},
	\end{equation*}
where $\Phi_{u,v}$, $\Phi_{u,v}$, $\Psi_u$, $\Psi_{u,v}$ and $\Psi_{u,v}$ denote the value of these parcels at the end of $\mathtt{T}$. If a node learns it is {\em not} blacklisted, then it can delete its own status parcels from the previous transmission it participated in. In terms of storing and transferring other nodes' testimony parcels: nodes only keep the most recent testimony parcels of the other nodes.

Finally, if the Receiver is able to decode the current codeword, or if the Receiver notices inconsistencies (see \eqref{potEqn} below) from the potential parcels $\{\Phi_u\}$, then the Receiver constructs a single alert parcel indicating this fact.  Once the Sender gets this parcel, he will end the current message transmission either as S1 or F2 (as appropriate). The signal that will alert the Receiver of potential inconsistencies is if the following equation is ever satisfied (based on the most recent potential parcels $\Phi_u$ that the Receiver has):
	\begin{equation} \label{potEqn}
	\sum_{u \in G} \Phi_u \thickspace > \thickspace kCD
	\end{equation}
\subsection{Analysis of the Routing Protocol}\label{lastAnal}
\indent \indent In this section we present proofs regarding the correctness of our protocol, its memory requirements, and its competitive-ratio with respect to throughput.
\begin{thm} \label{memoryThmPlus}
The protocol of Section \ref{lastDesc} requires $O(n^2P)$ bits of memory for internal nodes.
\end{thm}
\begin{proof}
Recall that $P \in \Omega(k^2)$ (for security parameter $k$), so that signatures (of size $O(k)$ bits) on all communication and the obfuscated count $E(\chi(p))$ (encrypted $k$-tuple, with each coordinate $k$ bits for the encryption, for a total of $O(k^2)$ bits) fit inside a single packet. A node must store up to $C = \Theta(n^2)$ codeword parcels of size $\Theta(P)$.  In terms of control information, at any time a node must store: At most $2n$ parcels from the {\em Sender/Receiver alerts}; at most $n$ {\em status parcels}; at most $n^2$ {\em testimony parcels}; and at most $n$ {\em potential parcels}.
\end{proof}
\begin{lmma}\label{threeTog} The protocol of Section \ref{lastDesc} satisfies the following three properties:
	\begin{enumerate}
	\item After a corrupt node is eliminated (or at the outset of the protocol) and before the next corrupt node is eliminated, there can be at most $n$-1 \hspace*{-.5pt}failed transmissions $\{\hspace*{-.5pt}\mathtt{T}_1, $...$, \mathtt{T}_{n\mbox{-}1}\hspace*{-1.5pt}\}$ \hspace*{-1pt}before there is an index $i \hspace*{-2pt}\in \hspace*{-4pt}[1,n]$ such that the Sender has all testimonies from nodes participating in $\mathtt{T}_i$.
	
	\item For any transmission that fails as in F2, the Sender can eliminate a corrupt node if it receives the testimony from every node that participated in that transmission.
	
	\item For any transmission that fails as in F3, (with all but negligible probability in the security parameter $k$) the Sender can eliminate a corrupt node if it receives the testimony from every node that participated in that transmission.


	\end{enumerate}
\end{lmma}
\begin{proof}
The intuition for Statement (1) is that the blacklist forces corrupt nodes to return their testimonies to the Sender if they want to further disrupt future transmissions.  In particular, nodes for which the Sender is missing a testimony parcel will not be allowed to transfer (codeword) parcels in future transmissions, as they will be on the blacklist. Therefore, we can associate to each failed transmission (at least) one {\em distinct} node for which the Sender is still missing a testimony parcel, where by ``distinct'' we mean this node cannot be associated to more than one transmission (since after a node is associated to a failed transmission, the Sender will not require its testimony for any transmissions {\em after} that one, as it will not have participated in those transmissions). Thus, by a simple counting argument there can be at most $n-1$ failed transmissions before the Sender necessarily has all testimonies from nodes participating in one of those failed transmissions. This argument is formalized in \cite{paper} (their proof remains valid here, as both protocols utilize the same principles of the blacklist).

Regarding Statement (2), notice that a transmission fails as in Case F2 when a corrupt node is transferring codeword parcels against transfer rules (e.g$.$ from {\em smaller} heights to {\em larger} heights, or when a corrupt node is duplicating codeword parcels).  Both of these can be detected via the {\em potential differences} portion of the nodes' testimonies. In particular, there will necessarily exist some node $u$ whose cumulative directional potential drop is greater for parcels {\em sent} than for parcels {\em received}: $\sum_{v \in G} \Phi_{u,v} \thickspace > \thickspace \sum_{v \in G} \Phi_{v,u}$. The existence of such a node $u$, and the fact that $u$ is corrupt, is proven rigorously in \cite{paper}; their proof remains valid here, as both protocols utilize the same potential difference parcels portion of the control information.

Statement (3) is the primary novel contribution our protocol achieves, and will be proven via a sequence of lemmas below.
\end{proof}
Before proving Statement (3), we demonstrate how Lemma \ref{threeTog} yields Theorem \ref{mtn}.
\begin{lmma}
\label{theyDont} Let $\mathcal{P}$ denote the protocol described in Section \ref{lastDesc}. For any adversary $\mathcal{A}$, for any protocol $\mathcal{P}'$, and for any $\mathtt{x} \in \mathbb{N}$, if at any time $\mathcal{P}'$ has received $\Theta(xn)$ messages, then $\mathcal{P}$ has received $\Omega(x-n^3)$ messages.
\end{lmma}
\begin{proof}
In the rounds that form any transmission for $\mathcal{P}$ (regardless of whether it ended as in S1 or F2-F4), any competing protocol $\mathcal{P}'$ can deliver at most $\Theta(kn^2C)$ packets. This fact is proven in \cite{paper}, and the proof remains valid here as the protocol rules for transferring {\em codeword parcels} is the same for both protocols (only the specific data stored in the control information is different).

There are at most $n^2$ {\em failed} transmissions F2-F4 (Statement (1) of Lemma \ref{threeTog}), and the above guarantees that in these transmissions, any competing protocol $\mathcal{P}'$ can deliver at most $O(n)$ messages (since messages consist of $\Theta(knC)$ packets).  Therefore, any competing protocol can deliver at most $O(n^3)$ messages during failed transmissions.  Note that the extra (at most) $O(n^3)$ messages that a competing protocol delivers does not affect competitive-ratio, as they can be absorbed in the additive term $g(n,C)$ for competitive-ratio (see Definition \ref{tdef}).

Meanwhile for all {\em successful} transmissions, our protocol delivers a message of $\Theta(knC)$ parcels, which as noted above is within a factor of $n$ of the throughput performance of any protocol $\mathcal{P}'$.
\end{proof}
\noindent {\em Proof of Theorem \ref{mtn}.} Recall that correctness means that the Receiver gets the (unaltered) messages sent by the Sender in-order.  The integrity of the messages received by the Receiver is assured (with all but negligible probability in the security parameter) by the fact that the Sender signs all messages, and no (honest) node (and in particular the Receiver) ever accepts a packet that does not have a valid signature.  Finally, error-correction (which allows for some codeword parcels to arrive out of order and/or get lost in transmission) together with the fact that the Sender repeats all codeword transmissions that failed as in Case F2-F4 ensures that the Receiver gets all messages in order. That our protocol achieves competitive-ratio $1/n$ follows immediately from Lemma \ref{theyDont}.\hspace*{\fill}$\blacksquare$\\[12pt]
\indent The proof of Statement (3) of Lemma \ref{threeTog} will rely on the novel use of the {\em obfuscated count} portion of the control information. The intuition for how the obfuscated counts can be used to identify a corrupt node performing {\em packet replacement} was given in Section \ref{ourResultz}. To see how packet replacement relates to failing as in F3, recall that Case F3 means that the Sender has inserted all $D$ codeword parcels, and yet the Receiver has not gotten $(1-\lambda)D$ of them (since otherwise the Receiver could have decoded the message and the transmission would have ended as in S1). Since there can be at most $(n-2)C$ parcels that are (validly) stored in the buffers of the internal nodes, parcel deletion/replacement must have occurred: $D > (1-\lambda)D + (n-2)C$. This is exactly the situation for which the {\em obfuscated count} allows the Sender to identify a corrupt node.

Let $s$ denote the Sender and $r$ denote the Receiver. As a simplification, we will drop the $\mathcal{E}$ from our notation (writing instead simply $\chi(p)$), since ultimately it is the Sender (who has the decryption key for $\mathcal{E}$) who will be evaluating these parcels. For a given message transmission, let $\mathcal{S}$ denote the set of codeword parcels that the Sender inserted to an internal node (i.e$.$ {\em not} to the Receiver). Let $\mathcal{R}$ be the multiset of codeword parcels that the Receiver received from the internal nodes (i.e$.$ {\em not} the Sender); note that $\mathcal{R}$ may be a {\em multi}set, as corrupt nodes may have duplicated codeword parcels.

Note that there is potential ambiguity in the notation $\Psi_{u,v}$: it can either represent the values $u$ has stored (with signatures from $v$) for codeword parcels sent from $u$ to $v$; or it can represent the values $v$ has stored (with signatures from $u$) for parcels sent from $u$ to $v$.  If $u$ and $v$ are honest, these quantities will match.  When discussing these quantities when they are both returned to the Sender via the testimonies, we use the following convention which removes ambiguity:
	\begin{enumerate}
	\item If the hamming distance between $u$ and $v$'s testimony for (the decrypted value of) $\Psi_{u,v}$ is greater than one, then the Sender can immediately eliminate a corrupt node.  After all, all nodes should be verifying the values sent by their neighbor are valid before communicating more codeword parcels with them (see Figure 1).  So if $u$ and $v$'s values for $\Psi_{u,v}$ have hamming distance greater than one, the Sender can necessarily identify the node that returned the less recent value for $\Psi_{u,v}$ (status parcels are time-stamped with the round).\footnote{If the time-stamps indicate the two values for $\Psi_{u,v}$ are from the {\em same} round, then both $u$ and $v$ are corrupt.}

	\item If the hamming distance between $u$ and $v$'s testimony for (the decrypted value of) $\Psi_{u,v}$ is exactly equal to one, then the Sender sets the more outdated value to the more current value, and also adjusts $\Psi_w$ of the appropriate node ($w \in \{u,v\}$) accordingly.  For example, if $\mathbf{v}$ is the difference in the two values for $\Psi_{u,v}$ and \hspace*{-1pt}$u$ has the more current\footnote{If $u$ and $v$'s returned values for $\Psi_{u,v}$ are both time stamped with the {\em same} round, then the Sender treats the value coming from the {\em receiving} node (in this case $v$) as the more current value.} timestamp, then the Sender sets $v$'s value for $\Psi_{u,v}$ equal to $u$'s value and modifies $\Psi_v$ by $\mathbf{v}$: $\Psi_v := \Psi_v + \mathbf{v}$ (thus, the quantity $\Psi_{u,v} + \Psi_v$ remains constant through these changes).
	\end{enumerate}

The following equalities come from the fact that every contribution to $\Psi_{s,u}$ (respectively $\Psi_{u,r}$) can be ascribed to a single codeword parcel that the Sender inserted (respectively, that the Receiver received); see Figure 1:
%
	\begin{equation}\label{bahBah}
	\sum_{u \in G \setminus \{r,s\}} \hspace*{-8pt}\Psi_{s,u} \thickspace = \sum_{p \hspace*{1pt}\in \hspace*{1pt} \mathcal{S}} \hspace*{-1pt}\chi(p) \quad \mbox{and} \quad \sum_{u \in G \setminus \{r,s\}} \hspace*{-8pt}\Psi_{u,r} \thickspace = \sum_{p \in \mathcal{R}} \hspace*{-1pt}\chi(p)
	\end{equation}
%
%
The following equation expresses the fact that for every codeword parcel an honest node receives, that codeword parcel will either have been transferred (exactly once) to an adjacent node, or the codeword parcel will still be in that node's buffer at the end of the transmission; see Figure 1:
	\begin{equation} \label{item3MMM}
	\mbox{For any uncorrupt node $u$:} \qquad (0, \dots, 0) \thickspace = \thickspace \Psi_u \thickspace + \thickspace \sum_{v \in G\setminus u}\left(\Psi_{u,v} - \Psi_{v,u}\right)
	\end{equation}
\begin{lmma}\label{faith} If the equality in \eqref{item3MMM} holds for all nodes $u \in G \setminus \{r,s\}$, then:
	\begin{equation}\label{dishes}
	\sum_{p \in \mathcal{S}} \chi(p) \thickspace = \hspace*{-5pt}\sum_{u \in G \setminus \{r,s\}} \hspace*{-8pt}\Psi_u \thickspace + \thickspace \sum_{p \in \mathcal{R}} \chi(p)
	\end{equation}
\end{lmma}
\begin{proof} This lemma simply states that the sum of the characteristic vectors $\chi(p)$ for each parcel inserted by the Sender will equal the sum of the characteristic vectors received by the Receiver plus the sum of the characteristic vectors of the parcels still stored by the internal nodes at the end of the transmission. Note that \eqref{dishes} will always hold if there are no corrupt nodes; but in the case a corrupt node is performing {\em packet replacement}, \eqref{dishes} will only hold if the set of parcels replaced have the same distribution within the sets (as determined by $\chi$) as the parcels they were replaced with.

Formally, consider:
	\begin{alignat}{2}
	\sum_{u \in G} \sum_{\hspace*{4pt}v \in G \setminus u} \left( \Psi_{u,v} - \Psi_{v,u} \right) \thickspace &= \thickspace (0, \dots, 0) \quad \Rightarrow \notag \\
	\sum_{u \in G \setminus \{r,s\}} \sum_{v \in G \setminus u} \hspace*{-5pt}\left( \Psi_{u,v} - \Psi_{v,u} \right) \thickspace &= \thickspace \hspace*{-7pt}\sum_{v \in G \setminus r} \hspace*{-5pt}\Psi_{v,r} - \hspace*{-4pt}\sum_{v \in G \setminus s} \hspace*{-4pt}\Psi_{s,v} \quad \hspace*{-1.4pt}\Rightarrow \notag \\
	\sum_{u \in G \setminus \{r,s\}} \sum_{v \in G \setminus u} \hspace*{-5pt}\left( \Psi_{u,v} - \Psi_{v,u} \right) \thickspace &= \thickspace \sum_{p \in \mathcal{R}} \chi(p) - \sum_{p \in \mathcal{S}} \chi(p), \label{galileoTX}
	\end{alignat}
where the first equality is by symmetry; the second comes from the first by separating terms and noting that $\Psi_{r,v} = \mathbf{0} = \Psi_{v,s}$ for all $v$ (since the Sender never {\em receives} a codeword parcel and the Receiver never {\em sends} a codeword parcel); and the third is \eqref{bahBah}.  Therefore:
	\begin{alignat}{2}
	&(0, \dots, 0) \thickspace = \thickspace \hspace*{-8pt}\sum_{u \in G \setminus \{r,s\}} \hspace*{-8pt}\Psi_u \thickspace + \hspace*{-8pt}\sum_{u \in G \setminus \{r,s\}}\sum_{v \in G \setminus u} \left( \Psi_{u,v} - \Psi_{v,u} \right) \quad \Rightarrow \notag \\
	&(0, \dots, 0) \thickspace = \thickspace \hspace*{-8pt}\sum_{u \in G \setminus \{r,s\}} \hspace*{-8pt}\Psi_u \thickspace + \thickspace \sum_{p \in \mathcal{R}} \chi(p) \thickspace - \thickspace \sum_{p \in \mathcal{S}} \chi(p) \notag
	\end{alignat}
where the first equality is by hypothesis (that equality \eqref{item3MMM} is true for all nodes $u \in G \setminus \{r,s\}$) and the second equality comes from the first via \eqref{galileoTX}.
\end{proof}
\begin{lmma} \label{eileenM} Suppose a transmission fails as in Case F3, and at some later point the Sender has collected all of the testimonies from all nodes participating in that transmission, and that equation \eqref{item3MMM} is satisfied for all nodes.  The probability that the following equality is satisfied is negligible in $k$:
	\begin{equation}\label{riverDM}
	\sum_{u \in G \setminus \{r, s\}} \hspace*{-9pt}\Psi_{s,u} \hspace*{3pt} \thickspace = \thickspace \hspace*{-5pt}\sum_{u \in G \setminus \{r, s\}} \hspace*{-6pt} \left(\Psi_u + \Psi_{u,r}\right)
	\end{equation}
More precisely, \eqref{riverDM} is satisfied with probability at most: $\frac{\sqrt{ek}}{\left(\sqrt{2\pi}\right)^{k-1}}$.
\end{lmma}
\begin{proof} Fix a transmission that satisfies the hypotheses of the lemma, and let $\{\Psi_u, \thickspace \Psi_{u,v}, \thickspace \Psi_{v,u} \}_{u,v \in G}$ denote the testimonies the Sender has collected. Let $\mathcal{R}_{\cap} := \mathcal{S} \cap \mathcal{R}$, i.e$.$ $\mathcal{R}_{\cap}$ represents a {\em set} (as opposed to a {\em multi}set) consisting of the same parcels as $\mathcal{R}$ (if there was no corrupt activity, then $\mathcal{R}_{\cap} = \mathcal{R}$).
Since (considered as multisets) $\mathcal{R}_{\cap} \subseteq \mathcal{R}$, by Lemma \ref{faith} there exists $Q \subseteq \mathcal{S}$ such that $\mathcal{R}_{\cap} \subseteq \mathcal{Q}$ and:
	\begin{equation} \label{bumBum}
	\sum_{p \in \mathcal{Q}} \chi(p) \thickspace = \hspace*{-4.5pt}\sum_{u \in G \setminus \{r, s\}} \hspace*{-8pt}\Psi_u \thickspace + \sum_{p \in \mathcal{R}_{\cap}}\chi(p)
	\end{equation}
Consequently:
	\begin{alignat}{2}
	\sum_{u \in G \setminus \{r, s\}} \hspace*{-9pt}\Psi_{s,u} \hspace*{3pt} \thickspace &= \thickspace \hspace*{-5pt}\sum_{u \in G \setminus \{r, s\}} \hspace*{-6pt} \left(\Psi_u + \Psi_{u,r}\right) \quad \Leftrightarrow \notag \\
	\sum_{p \in \mathcal{S}} \chi(p) \thickspace &= \hspace*{-4.5pt}\sum_{u \in G \setminus \{r, s\}} \hspace*{-8pt}\Psi_u \thickspace + \sum_{p \in \mathcal{R}}\chi(p) \quad \Leftrightarrow \notag \\
	\sum_{p \in \mathcal{S} \setminus \mathcal{Q}} \hspace*{-4pt}\chi(p) \thickspace &= \thickspace \hspace*{-7pt}\sum_{p \in \mathcal{R} \setminus \mathcal{R}_{\cap}} \hspace*{-7pt}\chi(p),\label{isignM}
	\end{alignat}
where the first implication is \eqref{bahBah} and the second is \eqref{bumBum}. The right-hand side of \eqref{isignM} represents the sum of the characteristic vectors $\chi(p)$ for the parcels $p$ that the Adversary duplicated (e.g$.$ via packet replacement). Meanwhile, the left-hand side of \eqref{isignM} is completely random, based on the choice of $\chi:D \rightarrow \mathbb{Z}_N^K$. Notice that $\varnothing = (\mathcal{S} \setminus \mathcal{Q}) \cap (\mathcal{R} \setminus \mathcal{R}_{\cap})$ by definition of $\mathcal{R}_{\cap}$ and choice of $Q$. The fact that these sets are disjoint means that the quantities on both sides of \eqref{isignM} are {\em independent} from each other.

Therefore, the probability that the equality in \eqref{isignM} holds is equal to the probability that a random assignment $\chi(p)$ for each $p \in \mathcal{S} \hspace*{-2pt}\setminus \hspace*{-2pt}\mathcal{Q}$ produces the vector described by the right-hand side of \eqref{isignM}. Notice that if $\mathcal{R} = \mathcal{R}_{\cap}$ (i.e$.$ no malicious packet replacement was performed), then $\mathcal{Q} = \mathcal{S}$, and \eqref{isignM} is vacuously satisfied. But if the Adversary attempts to invoke a packet replacement strategy, the probability that the Adversary's strategy will maintain equality in \eqref{isignM} is equivalent to the probability that the Adversary can correctly predict the outcome distribution of an experiment in which $m$ balls are distributed into $K$ buckets, where each of the $m = |\mathcal{S} \hspace*{-2pt}\setminus \hspace*{-2pt}\mathcal{Q}|$ balls are assigned a bucket uniformly at random. It is a straightforward argument (see Fact \ref{wildWest} for details) that demonstrates this probability is bounded above by\footnote{Notice that $|\mathcal{S} \setminus \mathcal{Q}|$ does not appear anywhere in $(\sqrt{eK})\left(\sqrt{2\pi}\right)^{1-K}$.  Although the exact probability does depend on $|\mathcal{S} \setminus \mathcal{Q}|$, the bound is valid as long as $|\mathcal{S} \setminus \mathcal{Q}| > K$ (see Fact \ref{wildWest}), which will happen since $K :=k$.} $\sqrt{eK\left(2\pi\right)^{1-K}}$, where we have used that $|\mathcal{S} \hspace*{-2pt}\setminus \hspace*{-2pt}\mathcal{Q}| > (k-1)nC$:
	\begin{equation*}
	|\mathcal{S} \hspace*{-2pt}\setminus \hspace*{-2pt}\mathcal{Q}| \hspace*{4pt}=\hspace*{4pt} |S| - |Q| \hspace*{4pt}=\hspace*{4pt} D - \hspace*{-10pt}\sum_{u \in G \setminus \{r,s\}} \hspace*{-10pt}\|\Psi_u\|_{_{\scriptstyle 1}} - |R_{\cap}| \hspace*{4pt}>\hspace*{4pt} D - (n-2)C - (1 - \lambda)D \hspace*{4pt}>\hspace*{4pt} (k-1)nC,
	\end{equation*}
where the second equality is via \eqref{bumBum} and because we are in Case F3 (Sender inserted all codeword parcels), the next inequality is because the Receiver could not decode (and hence got fewer than $(1-\lambda)D$ distinct codeword parcels) and by the fact that $\|\Psi_u\|_{_{\scriptstyle 1}} \leq C$ for all nodes $u$ (as otherwise the Sender could immediately eliminate $u$ as corrupt, since $u$ is storing more codeword parcels than allowed), and the final inequality is by construction: $D:=knC/\lambda$. Note that we are using the standard $l_1$-norm $\|\cdot\|_{_{\scriptstyle 1}}$ on $\mathbb{Z}^k_N$:
	\begin{equation*}
	\|(v_1, v_2, \dots, v_k)\|_{_{\scriptstyle 1}} = |v_1| + \dots + |v_k|,
	\end{equation*}
where $v_i \in [0..N-1]$ are the canonical representatives in $\mathbb{Z}_N$.
\end{proof}
\begin{cor}\label{kissesM} Suppose a transmission fails as in Case F3, and at some later point the Sender has collected all of the testimonies from all nodes participating in that transmission. Then with overwhelming probability, there will be a node $u \in G \setminus \{r, s\}$ that fails to satisfy \eqref{item3MMM}, and thus the Sender can eliminate $u$ as corrupt.
\end{cor}
\begin{proof} Suppose for the sake of contradiction that all nodes satisfy \eqref{item3MMM}. Then Lemma \ref{eileenM} states that \eqref{riverDM} will {\em not} be satisfied (with overwhelming probability), and so:
	\begin{equation}\label{timeyM}
	(0, \dots, 0) \thickspace \neq \hspace*{-4pt} \sum_{u \in G \setminus \{r,s\}} \hspace*{-4pt}\Psi_u \hspace*{4pt}+ \hspace*{-6pt}\sum_{u \in G \setminus \{r,s\}} \hspace*{-6pt}\left( \Psi_{u,r} - \Psi_{s,u} \right)
	\end{equation}
Also, by symmetry we have the trivial identity:
	\begin{equation}\label{timezM}
	(0, \dots, 0) = \sum_{u \in G \setminus \{r,s\}} \hspace*{-9pt}\sum_{\hspace*{15pt}v \in G \setminus \{u,r,s\}} \hspace*{-15pt}\left( \Psi_{u,v} - \Psi_{v,u} \right)
	\end{equation}
Combining \eqref{timeyM} and \eqref{timezM} and using the fact that $\Psi_{r,u} = \mathbf{0} = \Psi_{u,s}$ for all $u$ (since the Sender never {\em receives} a codeword parcel and the Receiver never {\em sends} a codeword parcel), with overwhelming probability:
	\begin{equation}\label{timexM}
	(0, \dots, 0) \thickspace \neq \hspace*{-9pt} \sum_{u \in G \setminus \{r,s\}} \hspace*{-7pt}\Psi_u \hspace*{4pt}+ \hspace*{-10pt} \sum_{u \in G \setminus \{r,s\}}\hspace*{-4pt}\sum_{\hspace*{6pt}v \in G \setminus u} \hspace*{-8pt}\left( \Psi_{u,v} - \Psi_{v,u} \right)
	\end{equation}
Therefore, there must be at least one node $u \in G \setminus \{r,s\}$ that violates \eqref{item3MMM}, which contradicts the hypothesis.
\end{proof}
Statement (3) of Lemma \ref{threeTog} follows immediately from Corollary \ref{kissesM}. We conclude by stating (and proving) some basic results from probability that led to the bound used in Lemma \ref{eileenM}.
\begin{fact} \label{probL2} For positive integers $m,k$ and a sequence of non-negative integers $\{n_i\}_{i=1}^k$ satisfying $m = \sum_{i=1}^k n_i$, the quantity $(n_1! n_2! \dots n_k!)$ is \textbf{minimized} if $|n_i - n_j| \leq 1$ for all pairs $(i, j)$.
\end{fact}
\begin{proof} (Proof by contradiction.)  Suppose that there is a sequence $\{n_1, \dots, n_k\}$ that achieves the minimal value for $(n_1! n_2! \dots n_k!)$ and yet does not satisfy the hypothesis that $|n_i - n_j|  \leq 1$ for all pairs $(i, j)$.  In particular, let $i$ and $j$ be such that $n_i - n_j \geq 2$, and let $M = n_1! n_2! \dots n_k!$ denote the minimal value obtained by this sequence $\{n_1, \dots, n_k\}$.  But this leads to a contradiction, since a lower minimum can be obtained from the sequence $\{n'_1, n'_2, \dots, n'_k\}$ satisfying: $n'_i = n_i -1$, $n'_j = n_j + 1$, and $n'_l = n_l$ for all other values of $l \in [1..k]$.
\end{proof}
\begin{fact} \label{probL1} Let $m, k, N$ be positive integers with $N > m \geq k$, and let $\mathsf{X}$ be a random variable defined by:
	\begin{equation}
	\mathsf{X} := \sum_{i=1}^m \mathbf{e}_{\mathsf{Y}_i},
	\end{equation}
where $\mathbf{e}_j \in \mathbb{Z}^k_N$ is the characteristic vector with a `1' in the $j^{th}$ position, and $\{\mathsf{Y}_i\}$ are independent random variables satisfying $Pr[Y_i = j] = 1/k$ (for any $j \in [1..k]$).  For any random vector $\mathbf{v} = (n_1, n_2, \dots, n_k)$ with $\sum_i n_i = m$, the probability $\mathsf{X} = \mathbf{v}$ obeys:
	\begin{alignat}{2}
	Pr[\mathsf{X} = \mathbf{v}] \thickspace &= \thickspace \frac{\left( \begin{array}{c} m \\ n_1  \end{array}\right)\left( \begin{array}{c} m-n_1 \\ n_2  \end{array}\right) \dots \left( \begin{array}{c} m - n_1 - n_2 \dots - n_{k-1} \\ n_k  \end{array}\right)}{k^m} \label{funProb} \\ 
	&= \thickspace \left(\frac{1}{k^m}\right)\frac{m!}{n_1! n_2! \dots n_k!} \notag
	\end{alignat}
\end{fact}
\begin{proof} Consider the following scenario: $m$ (unlabelled) balls are to be partitioned into $k$ labelled buckets, where each ball is assigned a bucket in a uniformly random manner.  At the end of this experiment, we may express the distribution of balls in buckets as the $k$-tuple: $(n_1, n_2, \dots, n_k) \in \mathbb{Z}_N^k$, where $n_i$ denotes the number of balls that ended up in bucket $i$.  It is clear that such an experiment describes the random variable $\mathsf{X}$, that is, for any $\mathbf{v} \in \mathbb{Z}_N^k$ with $m = \sum_{i=1}^k n_i$, we have that the above experiment results in the distribution $\mathbf{v}$ with the same probability that $\mathsf{X} = \mathbf{v}$.  Therefore, we prove the bound in \eqref{funProb} in terms of the scenario of partitioning $m$ (unlabelled) balls into $k$ (labelled) buckets.

Let $\mathbf{v} = (n_1, n_2, \dots, n_k)$ be fixed, and we determine the probability that the experiment will result in the distribution of balls in buckets as described by $\mathbf{v}$.  The fact that this probability is as described in \eqref{funProb} comes from a counting argument: the first term in the numerator is the number of ways $n_1$ balls can be chosen from $m$ balls to be assigned to the first bucket; the second term counts the number of ways $n_2$ balls can be chosen from the remaining $m-n_1$ balls to be assigned to the second bucket; and so on.  The denominator of the right-hand side of \eqref{funProb} counts the total number of ways $m$ balls can be distributed among $k$ buckets.  The second equality is a straightforward computation.
\end{proof}
The following fact formalizes the probability of correctly predicting the outcome of an experiment in which $m$ balls are distributed in $k$ buckets, where the choice of bucket for each ball is uniformly random. In particular, it states that the distribution with the highest probability is the one in which the $m$ balls are evenly distributed among the $k$ buckets (i.e$.$ each bucket has $m/k$ balls), and that the probability of this distribution is bounded above by $\frac{\sqrt{ek}}{\left(\sqrt{2\pi}\right)^{k-1}}$.
\begin{fact} \label{wildWest} Let $m, k, N$ be positive integers with $N > m \geq k$.  Let $W \subset \mathbb{Z}^k_N$ be a subset of vectors $\mathbf{v} \in \mathbb{Z}_N^k$ such that $\|\mathbf{v}\|_{_{\scriptstyle 1}} = m$.  Let $\mathsf{X}$ be a random variable as described in the hypothesis of Lemma \ref{probL1}.  For any fixed element $\mathbf{v} \in W$, the probability that $\mathsf{X} = \mathbf{v}$ is maximal if $\mathbf{v} = (m/k, m/k, \dots, m/k)$,\footnote{The hypotheses of the lemma do not assume $k | m$; when this is not the case, interpret $\mathbf{v} = (m/k, m/k, \dots, m/k)$ to mean $\mathbf{v} = (\lfloor m/k \rfloor, \dots, \lfloor m/k \rfloor, \lceil m/k \rceil, \dots, \lceil m/k \rceil)$, where the number of terms equal to $\lceil m/k \rceil$ is the remainder of $m$ divided by $k$.} and for this $\mathbf{v}$, the probability that $\mathsf{X} = \mathbf{v}$ is bounded from above by:
	\begin{equation}
	\frac{\sqrt{ek}}{\left(\sqrt{2\pi}\right)^{k-1}}
	\end{equation}
\end{fact}
\begin{proof} For any random vector $\mathbf{v} = (n_1, n_2, \dots, n_k)$ with $\sum_i n_i = m$, Fact \ref{probL1} demonstrates that the probability $\mathsf{X} = \mathbf{v}$ obeys:
	\begin{equation} \label{sists}
	Pr[\mathsf{X} = \mathbf{v}] \thickspace = \thickspace \left(\frac{1}{k^m}\right)\frac{m!}{n_1! n_2! \dots n_k!}
	\end{equation}
By Fact \ref{probL2}, \eqref{sists} is maximized if $|n_i - n_j| \leq 1$ for all pairs $(i, j)$, in which case $n_i \geq \lfloor m/k \rfloor$ for all $i$, and then plugging this into \eqref{sists}:
	\begin{alignat}{2}
	Pr[\mathsf{X} = \mathbf{v}] \thickspace &= \thickspace \left(\frac{1}{k^m}\right)\frac{m!}{n_1! n_2! \dots n_k!} \notag \\
	&\leq \thickspace \left(\frac{1}{k^m}\right)\frac{m!}{\left(\lfloor \frac{m}{k} \rfloor !\right)^k} \notag \\
	&\leq \thickspace \left(\frac{1}{k^m}\right)\frac{\sqrt{2e \pi m} \left( \frac{m}{e} \right)^m}{\left(\sqrt{2 \pi \lfloor \frac{m}{k} \rfloor} \left(\lfloor \frac{m}{k} \rfloor/e \right)^{\lfloor \frac{m}{k} \rfloor}\right)^k}, \label{toBecome}
	\end{alignat}
where the final inequality is via Stirling's Formula.  Since the quantity on the right side of \eqref{toBecome} is monotone decreasing as $m$ increases (for fixed $k \geq 1$), the probability is maximal for $m=k$ (given the constraint $m \geq k$), and then \eqref{toBecome} becomes:
	\begin{equation}
	Pr[\mathsf{X} = \mathbf{v}] \thickspace \leq \thickspace \frac{\sqrt{ek}}{\left(\sqrt{2\pi}\right)^{k-1}},
	\end{equation}
as claimed.
\end{proof}
The equalities of \eqref{bahBah} and \eqref{item3MMM} and the proofs of Lemmas \ref{faith}, \ref{eileenM}, and Corollary \ref{kissesM} relied on the fact that the coordinates of $\Psi_{u,v}$ never ``wrapped-around,'' i.e$.$ no coordinate ever became as large as $N$; in particular, we assumed addition on $\mathbb{Z}^k$ rather than $\mathbb{Z}^k_N$. The following Lemma assures us that if the Sender ever needs to use the $\{\Psi_{u,v}\}$ values to eliminate a corrupt node, then the values of the $\Psi_{u,v}$ are the same as if they were in $\mathbb{Z}^k$, i.e$.$ no coordinate ever reset as a result of getting as large as $N$.
\begin{lmma}\label{hurtsI} For any honest node $u \in G$, if at any time there exists any node $v \in G$ such that a coordinate of $\Psi_{u,v}$ (or $\Psi_{u,v}$) is $N-1$, then necessarily a corrupt node can be identified.
\end{lmma}
\begin{proof} We mentioned in Section \ref{Model} that $N := |\mathcal{G}| \in \Omega(kn^4)$, and more precisely, we demand $N > 3nD$. In this case, if any honest node has exchanged $N$ codeword parcels with a neighbor, then these $N$ transfers (each of which causes a potential drop of at least $|H'-H| \geq C/2n-2n$) in aggregate will correspond to a potential drop of at least $CD$, since $N > 3nD$ implies $N(C/2n - 2n) \geq CD$ (this uses the fact that $C \geq 12n^2$). Therefore, if any honest node has exchanged $N$ codeword parcels with a neighbor, then a corrupt node can be identified as is done for Case F2 of transmission failure.
\end{proof}

\end{document}